\documentclass[letterpaper, 10 pt, conference]{ieeeconf}

\IEEEoverridecommandlockouts                              

\usepackage{amsmath} 
\usepackage{amssymb}  

\usepackage{amssymb}
\usepackage{amsmath}
\usepackage{color}
\usepackage{graphicx}
\usepackage{dsfont}
\usepackage{mathrsfs}
\usepackage{etoolbox}
\usepackage{float}
\usepackage{subfigure}
\usepackage{amsbsy}
\usepackage{mathabx}
\usepackage{comment}
\usepackage{hyperref}
\usepackage{algorithm}
\usepackage{algpseudocode}
\usepackage{caption}
\usepackage{relsize}
\usepackage{siunitx}
\usepackage{soul}
\usepackage{cite}
\usepackage{textcomp}
\usepackage{multirow}
\usepackage[normalem]{ulem}
\usepackage{todonotes}
\definecolor{modcol}{RGB}{255,102,178}

\newcommand{\rev}[1]{\textcolor{black}{#1}}

\definecolor{MFabg}{RGB}{76, 191, 229}

\renewcommand{\emptyset}{\varnothing} 
\newcommand{\Rset}[1]{\mathbb{R}^{#1}} 

\def \p {\mathbf{p}}

\def \uu {\mathbf{u}}

\def \x {\mathbf{x}}

\def \H {\mathbf{H}}

\def \L {\mathbf{L}}

\def\Amc{\mathcal{A}}

\def\Cmc{\mathcal{C}}

\def\Emc{\mathcal{E}}

\def\Gmc{\mathcal{G}}

\def\Nmc{\mathcal{N}}

\def\Rmc{\mathcal{R}}
\def\Smc{\mathcal{S}}
\def\Tmc{\mathcal{T}}
\def\Umc{\mathcal{U}}
\def\Vmc{\mathcal{V}}
\def\Wmc{\mathcal{W}}

\newcommand{\Bomega}{\boldsymbol{\omega}}

\newcommand{\BDelta}{\boldsymbol{\Delta}}

\newtheorem{defn}{Definition}[section]
\newtheorem{asm}{Assumption}[section]
\newtheorem{thm}{Theorem}[section]
\newtheorem{lem}{Lemma}[section]
\newtheorem{prop}{Proposition}[section]
\newtheorem{rmk}{Remark}[section]

\newtheorem{prb}{Problem}[section]
\newtheorem{cor}{Corollary}[section]

\newcommand{\angb}[1] {\left<{#1}\right>}

\def \ones			{{\mathds{1}}} 
\newcommand{\onesvec}[1] 	{\ones_{#1}} 
\newcommand{\zerovec}[1] 	{\mathbf{0}_{#1}} 

\DeclareMathOperator{\diag}{diag} 

\newcommand{\eye}[1]{I_{#1}} 
\renewcommand{\vec}[3]{\mathrm{vec}_{#1}^{#2}(#3)} 
\renewcommand{\diag}[3]{\mathrm{diag}_{#1}^{#2}(#3)} 

\def \kronecker {\otimes} 

\def \graph		{\mathcal{G}} 

\title{A Robustness Analysis to Structured Channel Tampering \\ over Secure-by-design Consensus Networks}

\author{Marco~Fabris and~Daniel~Zelazo~\IEEEmembership{Senior Member,~IEEE}	
	\thanks{M. Fabris is with the Department of Information Engineering, University of Padua, Padua, Italy. 
		    D. Zelazo is with the Faculty of Aerospace Engineering, Technion, Haifa, Israel. This research was made possible thanks to the support of Ms. Beverly Bavil, Mr. David Dibner and another anonymous donor. 
		    \mbox{Corresponding author: D. Zelazo, {\tt \scriptsize \href{mailto:dzelazo@technion.ac.il}{dzelazo@technion.ac.il}}}}}

\begin{document}

\maketitle
\thispagestyle{empty}
\pagestyle{empty}

\begin{abstract}
	This work addresses multi-agent consensus networks where adverse attackers affect the convergence performances of the protocol by manipulating the edge weights. We generalize \cite{FabrisZelazo2022} and provide guarantees on the agents' agreement in the presence of attacks on multiple links in the network. A stability analysis is conducted to show the robustness to channel tampering in the scenario where part of the codeword, corresponding to the value of the edge weights, is corrupted. Exploiting the built-in objective coding, we show how to compensate the conservatism that may emerge because of multiple threats in exchange for higher encryption capabilities. Numerical examples related to semi-autonomous networks are provided.
\end{abstract}



\section{Introduction} 
\label{sec:intro}

The consensus problem, consisting in the design of networked control algorithms under which all individuals of a given multi-agent system (MAS) attain an agreement on a certain quantity of interest \cite{OlfatiFaxMurray2007}, is commonly tackled when it is required to achieve a global prefixed task.
However, due to the openness of communication protocols and the complexity of networks, the agreement of MASs may be vulnerable to malicious cyber-attacks \cite{MengXiaoLi2020}.
In particular, if the agent sensors are threatened by an attacker, the measured data may be unreliable or faulty. Indeed, the attack signals can even disrupt the control performance of the group of agents through the communication topology. 
Therefore, resilient solutions are required to ensure that MASs fulfill consensus under security hazards \cite{WangDengGuo2023}. Consequently, the secure control of MASs is now a crucial issue to be investigated \cite{ZhangFengShi2021,YangXiaoYang2021,HuoWuZhang2022}. 

Several recent studies illustrate the importance of giving guarantees against  cyber-threats while these are attempting to disrupt a MAS that is trying to reach consensus. In \cite{LiuShiYan2022,GaoHuJiang2022}, 
deception attackers injecting false data are assumed to attack the agents or communication channels. 
To counteract this kind of disruption, classic observers, impulsive control methods and event-triggered adaptive cognitive control have been leveraged. 
Denial of service (DoS) attacks then represent another challenging class of cyber-threats: robust control techniques have been developed in \cite{ElkhinderSiddigEl-Ferik2022,WangLiDuan2022,SathishkumarLiu2022} 
	to ensure sufficient levels of agreement over MASs under DoS 
	providing guarantees based on the maximum \textquotedblleft quality of service\textquotedblright~or Lyapunov theory. 

As already introduced in \cite{FabrisZelazo2022}, we embrace a different perspective. Instead of developing tools to secure existing networks (see \cite{GaoHeDong2022}), we provide inherently secure embedded measures through the adoption of a network manager to guarantee robust consensus convergence. \rev{For privacy and safety concerns, such a manager is not allowed to access local states. Rather, it only intervenes in an initial phase to ensure desired convergence performance via edge weight assignment (see e.g. \cite{XiaoBoyd2004}) in a secure way.} 
Nonetheless, differently from our previous work, this paper is meant to generalize the \textit{secure-by-design consensus} (SBDC) dynamics towards multi-agent consensus networks where adverse attackers affect the convergence performance \rev{through a \textit{structural} hit to the communication between manager and agents, thus corrupting edge weights in \textit{more than a single link} of the system.}  We summarize our main contributions below. 
 
\begin{itemize}
\item Two new guarantees based on the small gain theorem for the robust stability of the agents' agreement are given in both continuous and discrete time domains when multiple network edges undergo a structured weight deviation from their nominal values.

\item We introduce the notion of a resilience gap used to characterize the conservatism of the robustness analysis.  We show that for spanning trees, the resilience gap is always zero even in multi-attack scenarios.  We also discuss how this gap can be reduced by modulating the encryption capabilities used in the network.
\end{itemize}
The organization of the paper follows.
Sec. \ref{sec:prelim_models} introduces the preliminary notions for multi-agent consensus and the SBDC networks. 
Sec. \ref{sec:robust-analysis} provides a robustness analysis for the latter when subject to channel tampering modeled as multiple edge weight perturbations. 
A numerical case study on semi-autonomous networks is reported in Sec. \ref{sec:simulations} to assess such theoretical results and, lastly, concluding remarks and future works are sketched in Sec. \ref{sec:conclusions}.

\paragraph*{Notation} The set of real numbers, the $l$-dimensional (column) vector whose elements are all ones and the $l$-dimensional (column) null vector are denoted by $\Rset{}$, $\onesvec{l} \in \Rset{l}$ and $\zerovec{l}\in \Rset{l}$, respectively, while $\eye{l} \in \Rset{l\times l}$ refers to as the identity matrix. 
Let $\Omega \in \Rset{l \times l}$ be a square matrix. Relation $\Omega \succeq 0$ means that $\Omega$ is symmetric positive semidefinite. Notation $[\Omega]_{ij}$ addresses the entry of matrix $\Omega$ at row $i$ and column $j$, while $\Omega^{\top}$,  
$\mathrm{tr}(\Omega)$ and $\left\|\Omega\right\|$ denote respectively its transpose, 
its trace and its spectral norm. Operator $\mathrm{col}_{l}[\Omega]$ represents the $l$-th column of $\Omega$ and its $i$-th eigenvalue is denoted by $\lambda_{i}^{\Omega}$. 
The vector space spanned by a vector $\omega \in \Rset{l}$ is denoted with $\angb{\omega}$. The infinity norm of $\omega$ is identified by $\left\|\omega\right\|_{\infty}$.
Also, $\Bomega=\vec{i=1}{l}{\omega_{i}}$ defines the vectorization operator that stacks vectors $\omega_{i}$, $i=1,\dots,l$ as $\Bomega = \begin{bmatrix}
	\omega_{1}^{\top} & \dots & \omega_{l}^{\top}
\end{bmatrix}^{\top}$; whereas, $\diag{i=1}{l}{\varsigma_i}$ is a block diagonal matrix made up with $\varsigma_i \in \mathbb{R}$, $i=1,\dots,l$, on the diagonal and $\mathrm{diag}(\Bomega) = \diag{i=1}{l}{\omega_i}$.
Lastly, 
the Kronecker product is denoted with $\kronecker$.


\section{Preliminary results}
\label{sec:prelim_models}

Consensus models for MASs and preliminary notions are here given along with a brief overview of robustness results in consensus networks having multiple uncertain edges. The SBDC protocol is then briefly recalled.

\subsection{Overview on uncertain consensus networks}
An $n$-agent network can be modeled through a weighted graph $\mathcal{G}=\left(\mathcal{V},\mathcal{E},\Wmc\right)$ so that each element in the \textit{vertex set} $\mathcal{V}=\left\{1, \dots, n\right\}$ is associated to an agent in the group, while the \textit{edge set} $\mathcal{E} = \{e_{k}\}_{k=1}^{m}  \subseteq \mathcal{V}\times \mathcal{V}$ characterizes the agents' information exchange. 
In addition, $\Wmc = \{w_{k}\}_{k=1}^{m}$, with $m = |\Emc|$, denotes the set of weights attributed to each link. Throughout this work, bidirectional interactions among agents are assumed; hence, $\mathcal{G}$ is an \textit{undirected} graph.
A \textit{cycle} is defined as a non-empty and distinct sequence of edges joining a sequence of vertices, in which only the first and last vertices are equal. If a graph does not contain cycles it is said \textit{acyclic} and if it is also connected it is called a \textit{tree}.
The \textit{neighborhood} of node $i$ is defined as the set $\mathcal{N}_i=\left\{j\in\mathcal{V}\setminus\{i\} \;|\;(i,j)\in \Emc \right\}$, 
while the degree of node $i$ is defined through the cardinality $d_{i} = |\mathcal{N}_{i}|$ of neighborhood $\Nmc_i$.
Moreover, the \textit{incidence matrix} is denoted as $E \in \Rset{n\times m}$, in which each column $k \in \{1,\ldots,m\}$ is defined via the $k$-th (ordered) edge $(i,j) \in \Emc$, where $i<j$ is chosen w.l.o.g., and for edge $k$ corresponding to $(i,j)$ one has $[E]_{lk} = -1$, if $l = i$; $[E]_{lk} = 1$, if $l = j$; $[E]_{lk} = 0$, otherwise.
For all $k = 1,\ldots,m$, the weight $w_{k} := w_{ij} = w_{ji} \in \Rset{}$ is related to $k$-th edge $(i,j)$, and $W=\diag{k=1}{m}{w_k}$ is the diagonal matrix of edge weights.
Additionally, the \textit{Laplacian matrix}, incorporating the topological information about $\graph$, is denoted as $L(\Gmc) = E W E^{\top} $ (see \cite{Lunze2019}). Henceforward, we also suppose that graph $\graph$ is \textit{connected} and $L(\Gmc) \succeq 0$, thus having eigenvalues $\lambda_{i}^{L}$, for $i = 1,\ldots,n$, such that $0 = \lambda_{1}^{L} < \lambda_{2}^{L} \leq \cdots \leq \lambda_{n}^{L}$. A sufficient condition to satisfy the latter requisite, which is adopted throughout this paper, is to take $w_{ij}> 0$ for all $(i,j)$.
With an appropriate labeling of the edges, we can always assume that the incidence matrix $E = \begin{bmatrix}
	E_{\Tmc} & E_{\Cmc}
\end{bmatrix}$ can be partitioned into the incidence matrix $E_{\Tmc}$, relative to a spanning tree $\Tmc \subseteq \Gmc$ with $\tau=n-1$ edges, and the incidence matrix $E_{\Cmc}$, associated to $\Cmc = \Gmc \setminus \Tmc$. Consequently, we define the cut-set matrix of $\Gmc$ (see \cite{ZelazoBurger2017}) as
$R_{(\Tmc,\Cmc)} = \begin{bmatrix}
	I_{\tau} &  T_{(\Tmc,\Cmc)} 
\end{bmatrix}$, with $T_{(\Tmc,\Cmc)} = ( E_{\Tmc}^{\top} E_{\Tmc} )^{-1} E_{\Tmc}^{\top}E_{\Cmc}$. 

A summary of the weighted consensus problem in MASs follows.  Let us consider a group of $n$ homogeneous agents, modeled by a weighted and connected graph $\graph$, 
and assign a continuous-time state $x_{i} = x_{i}(t) \in \Rset{D}$ to the $i$-th agent, for $i = 1,\dots,n$. The full network state is given by $\x = \vec{i=1}{n}{x_i} 
\in X \subseteq \Rset{N}$, with $N=nD$. The weighted consensus  for a MAS can be characterized as follows.

\begin{defn}[Weighted Consensus \cite{Lunze2019}]\label{def:consensus}
	{An $n$-agent network achieves \emph{consensus} if $ \underset{t\rightarrow +\infty}{\lim}~ \x(t) \in \Amc $, where, for some $\omega \in \mathbb{R}^{D}$, $\Amc = \angb{\ones_{n}} \otimes \omega$ is termed the \textit{agreement set}.
	}
\end{defn}

For a connected graph $\Gmc$ with positive weights, it is well known that the \textit{linear weighted consensus protocol}, given by
\begin{equation}\label{eq:LAP}
	\dot{\x} = - \L(\Gmc) \x,
\end{equation}
where $\L(\Gmc) = (L(\Gmc) \kronecker I_{D})$, satisfies $\x(t) \in \Amc$ as $t \rightarrow +\infty$.

In this direction, we also revisit a robustness result for the consensus protocol with small-magnitude perturbations on the edge weights \cite{ZelazoBurger2017}. 	Within this framework, we take into consideration the perturbed Laplacian matrix $L(\Gmc_{\Delta^W}) = E (W+\Delta^W) E^{\top}$ for a structured norm-bounded perturbation 
\begin{align}\label{uncertaintyset}
    \Delta^W \in \BDelta^{W} = \{\Delta^W  \, : \, \Delta^W =\diag{k=1}{m}{\delta^{w}_k}, \|\Delta^W \| \leq \bar{\delta}^W  \}.
\end{align}
Letting $\Emc_{\Delta} := \{e_{1}^{\Delta},e_{2}^{\Delta},\ldots\} \subseteq \Emc$ be the (nonempty) subset of uncertain edges, we can define the matrix $P \in \{0,1\}^{|\Emc| \times |\Emc_{\Delta}|}$ that selects the uncertain edges in $\Emc$, with $[P]_{ij}=1$, if $i$ and $j$ satisfy $e_{i} = e_{j}^{\Delta}  $; and $[P]_{ij}=0$, otherwise. This~leads~to

\begin{lem} \label{thm:effective_resistance_m}
	Let $\Delta^{W} \in \mathbb{R}^{|\Emc_{\Delta}|\times|\Emc_{\Delta}|}$, with $\Delta^W$ diagonal, and consider the nominal weighted consensus protocol \eqref{eq:LAP}. Then the perturbed consensus protocol 
		\begin{equation}\label{eq:perturbed_cons_dyn}
		\dot{\x} = -(L(\Gmc_{\Delta^W}) \kronecker I_{D}) \x
	\end{equation} 
	achieves consensus $\forall\Delta^{W} \in \BDelta^{W}$ (defined in \eqref{uncertaintyset}), if
	\begin{equation}\label{eq:uncguar_multi}
		\hspace{-.2cm}\left\|\Delta^{W} \right\| < {\underbrace{ \left\| P^{\top} R_{(\Tmc,\Cmc)}^{\top} (R_{(\Tmc,\Cmc)} W R_{(\Tmc,\Cmc)}^{\top})^{-1} R_{(\Tmc,\Cmc)} P  \right\|}_{=:\Rmc_{\Emc_{\Delta}}(\Gmc)}}^{-1}. 
	\end{equation}

\end{lem}
\begin{proof}
		By \cite[Theorem V.2]{ZelazoBurger2017}, the edge-agreement version of \eqref{eq:perturbed_cons_dyn} is asymptotically stable. Consequently, as $\Gmc_{\Delta^W}$ is connected, one has $L(\Gmc_{\Delta^W}) \succeq 0$ with a simple eigenvalue at $0$. Therefore, Def. \ref{def:consensus} is satisfied $\forall\Delta^{W} \in \BDelta^{W}$.
\end{proof}

In a slight abuse of convention, we refer to Lemma \ref{thm:effective_resistance_m} as \emph{robust stability} result; see \cite{ZelazoBurger2017} for more discussion on this notion.  
When $|\Emc_{\Delta}|=1$, it was shown in \cite{ZelazoBurger2017} that the bound in \eqref{eq:uncguar_multi} is tight and that $\Rmc_{\left\lbrace (u,v) \right\rbrace}(\Gmc)$ can be interpreted as the \textit{effective resistance} 
between a pair of nodes $(u,v)$.  However, for multi-edge attacks this bound is inherently conservative; this aspect will be elaborated upon in Section \ref{conserv_analysis}.


\subsection{Secure-by-design consensus dynamics}\label{ssec:SBDCdyn}
In this work, we consider MASs that are subject to the same key principles assumed in \cite{FabrisZelazo2022}, that is with presence of \rev{tasks, objective coding, information localization and a network manager}\footnote{\rev{The network manager does not govern the agents' dynamics, namely it should not be intended to fulfill the role of a global controller. Instead, it is chosen to precisely serve as an encryption mechanism to set up and secure distributed algorithms running on the underlying MAS.}}. 
These elements, which comprise the basic setup of the SBCD dynamics, are recalled in the following lines.

Tasks are described by an encoded parameter $\theta$ that we term the \textit{codeword} and the space of all tasks is denoted as $\Theta$. Each agent in the network then decodes this objective using its \textit{objective decoding function}, defined as $p_i: \Theta \to \Pi_i$, where $\Pi_i$ depends on the particular application (e.g. $\Pi_i \subseteq \Rset{n}$ within the consensus setting). For $\theta \in \Theta$, $p_i(\theta)$ is called the \emph{localized objective}. Instead, if $\theta \notin \Theta$, $p_i(\theta)$ may not be computable; nonetheless, any agent collecting such a codeword may send an alert. More precisely, the objective coding is established via the non-constant functions $p_{i}(\theta) : \Theta \rightarrow \Pi_i \subseteq \Rset{n}$, such that $ [p_{i}(\theta)]_{j} := p_{ij}(\theta) $ with 
\begin{equation}\label{eq:pijbasicchar}
	p_{ij}(\theta) = \begin{cases}
		w_{ij}, \quad \text{ if } (i,j) \in \Emc; \\
		0, \quad~~~\!  \text{ otherwise}.
	\end{cases}
\end{equation}
\rev{The values $w_{ij}$ in \eqref{eq:pijbasicchar} coincide with the nominal desired consensus weights that are assigned, encoded and broadcast by the network manager, to achieve desired convergence rates or other performance metrics in consensus networks.} 
Moreover, the information localization about the global state $\x$ is expressed by means of $h_{i}(\x) : X \rightarrow  Y_{i} \subseteq \Rset{D \times n} $, such that $\mathrm{col}_{j}[h_{i}(\x)] := h_{ij}(\x(t)) \in \Rset{D}$ with $h_{ij}(\x)=x_{i}-x_{j}$, if $(i,j) \in \Emc$; $h_{ij}(\x)=\zerovec{D}$, otherwise. 
So, the $i$-th agent's dynamics for the SBDC is determined by
\begin{align}\label{eq:secure_consensus} 
	\dot{x}_{i} = -  {\textstyle\sum}_{j \in \Nmc_{i}} p_{ij}(\theta) h_{ij}(\x)   
	. 
\end{align}

Remarkably, \eqref{eq:secure_consensus} reproduces exactly the linear consensus protocol already introduced in \eqref{eq:LAP}
and, defining $\p(\theta) = \vec{i=1}{n}{p_i(\theta)} \in \Rset{n^{2}}$ and $\H(\x) = \diag{i=1}{n}{h_i(\x(t))} \in \Rset{N \times n^{2}}$, dynamics \eqref{eq:secure_consensus} can be rewritten as
\begin{equation}\label{eq:SBDC}
	\dot{\x} = -\H(\x) \p(\theta),
\end{equation}
leading to the following theoretical result.
\begin{lem}[{\cite[Lemma III.1]{FabrisZelazo2022}}]\label{lem:standardandsecurepersp}
	The SBDC protocol \eqref{eq:SBDC} reaches agreement for any given objective decoding function $\p$ that satisfies \eqref{eq:pijbasicchar}.
\end{lem}

Likewise, we consider and investigate the well-known discrete-time consensus dynamics described by
\begin{equation}\label{eq:Consdt}
	\x(t+1) = \x(t) - \epsilon \L(\Gmc)\x(t)  ,
\end{equation}
where $\epsilon$ is a common parameter shared among all agents and designed to belong to the interval $(0,2/\lambda_{n}^{L})$, as shown in \cite{FabrisMichielettoCenedese2022}. According to the characterization in \eqref{eq:SBDC}, the SBCD dynamics adopted in discrete time in \eqref{eq:Consdt} is also given by
\begin{equation}\label{eq:SBDCdt}
	\x(t+1) = \x(t) - \epsilon \H(\x(t)) \p(\theta),
\end{equation}
since $\H(\x) \p(\theta) = \L(\Gmc)\x$ holds thanks to Lem. \ref{lem:standardandsecurepersp}.


\section{Robustness of the SBDC protocol to structured channel tampering} \label{sec:robust-analysis}

The original contribution provided by this study aims at the design of secure network systems to structured channel tampering while achieving consensus task. To this aim, the system is embedded with security measures that allow to render the network robust to small signal perturbations on some of the links. Also, a description for the structured channel tampering is given along with the relative robustness analysis for the SBDC protocol under multiple threats.

\subsection{Models, problem statement and key assumptions} \label{ssec:modelchanneltampmulti}

The structured channel tampering problem under analysis is formulated as follows.
Similarly to the model adopted in \cite{FabrisZelazo2022}, the prescribed codeword $\theta$ is subject to a perturbation $\delta^\theta \in \BDelta^{\theta} = \{\delta^\theta \, : \,  \|\delta^\theta \|_{\infty} \leq \bar{\delta}^{\theta}  \}$. We let $\Theta$ be a Euclidean subspace, in particular $\Theta = \Theta_{11} \times \Theta_{12} \times \cdots \times \Theta_{nn} \subseteq \Rset{n^{2}}$, and allow a codeword $\theta  = \vec{i=1}{n}{\theta_i} \in \Theta$ to be divided into (at most) $n(n-1)/2$ relevant ``subcodewords" $\theta^{(k)} := [\theta_i]_j = \theta_{ij}$, with $k=1,\ldots,m$, such that $\theta_{ij}=\theta_{ji}$, if $i\neq j$, and $\theta_{ii}$ is free to vary, for $i=1,\ldots,n$. Each $\theta_{ij} \in \Theta_{ij}  \subseteq \Rset{}$ can be seen as the $j$-th component of the $i$-th codeword fragment $\theta_i$, with $i=1,\ldots,n$. Such subcodewords influence the value of $p_{ij}(\theta)$ directly if and only if $j \in \Nmc_i$, i.e., it holds that $p_{ij}(\theta) = p_{ij}(\theta_{ij})$ for all $(i,j) \in \Emc$.


Therefore, 
the consensus description to support this investigation is such that the $i$-th nominal dynamics in \eqref{eq:SBDC} is modified as
\begin{equation}\label{eq:perturbedSBDC}
	\dot{x}_{i} = - {\textstyle\sum}_{j \in \Nmc_{i}} p_{ij}(\theta_{ij} + \delta^{\theta}_{ij}) h_{ij}(\x), \quad i = 1,\ldots,n,
\end{equation}
with $\delta^{\theta}_{ij} = [\delta^{\theta}_i]_j $ and $\delta^{\theta}_{i}$ satisfying $\delta^{\theta}  = \vec{i=1}{n}{\delta^{\theta}_i}$. 
Analogously, the $i$-th perturbed discrete time dynamics in \eqref{eq:SBDCdt}
can be written as:
\begin{equation}\label{eq:perturbedSBDCdt}
	x_{i}(t+1) = x_{i}(t) - \epsilon {\textstyle\sum}_{j \in \Nmc_{i}} p_{ij}(\theta_{ij} + \delta^{\theta}_{ij}) h_{ij}(\x(t)),
\end{equation}
where $\epsilon>0$ must be selected. In light of the previous discussion, the following design problem can be now stated.
\begin{prb}\label{prob:resilenceSBDC}
	Design objective functions $p_{ij}$ such that \eqref{eq:perturbedSBDC} (resp., \eqref{eq:perturbedSBDCdt} in discrete time) reaches consensus, independently from the codeword $\theta \in \Theta \subseteq \Rset{n^{2}}$, while the underlying MAS is subject to a structured attack $\delta^{\theta} \in \BDelta^{\theta}$ on multiple edges (belonging to $\Emc_{\Delta}$), i.e., with $\delta_{ij}^{\theta} = 0$ for all $(i,j)\in \Emc \setminus \Emc_{\Delta}$. 
	In addition, determine the largest $\rho^{\theta}_{\Delta}$ ensuring robust stability of \eqref{eq:perturbedSBDC} (resp., \eqref{eq:perturbedSBDCdt} in discrete time).
\end{prb}

Within this framework, it is possible to leverage Lem. \ref{thm:effective_resistance_m} and yield the main theoretical contribution of this paper, represented by the robustness guarantees for system \eqref{eq:perturbedSBDC} when the target of a cyber-physical attack is a multitude of edges. 
In this direction, we study how the robust stability of \eqref{eq:perturbedSBDC} is affected by perturbations on all the weights $p_{uv}(\theta_{uv}) = w_{uv}$ attached to the connections $(u,v) \in \Emc_{\Delta}$ that are caused by the deviations 
of each subcodeword $\theta_{uv}$.


As clarified later in more detail, the same three additional assumptions on the $p_{i}$'s proposed in \cite{FabrisZelazo2022} are \textit{sufficient} to tackle Problem \ref{prob:resilenceSBDC}, namely, this robustness analysis is again restricted to a particular choice for the objective coding, that is for concave and Lipschitz continuous differentiable functions $p_{i}$. Thus, we let
the $i$-th objective decoding function 
adopted in model \eqref{eq:perturbedSBDC} have the following properties.
\begin{asm} \label{asm:3properties}
	Each $p_{i}: \Theta \rightarrow \Pi_{i}$, with $i = 1,\ldots,n$, has the subsequent characterization:
	\begin{enumerate}
		\item[\textit{(i)}] values $[p_{i}(\theta)]_{j} = p_{ij}(\theta_{ij})$, with $\theta_{ij} = [\theta_i]_j$, satisfy \eqref{eq:pijbasicchar} for all $(i,j) \in \Emc$ and  are not constant w.r.t. $\theta_{ij}$; 
		\item[\textit{(ii)}] $p_{ij}$ is concave $\forall \theta \in \Theta$, i.e., $p_{ij}(\varsigma \eta_{1} + (1-\varsigma) \eta_{2}) \geq \varsigma p_{ij}(\eta_{1}) + (1-\varsigma) p_{ij}(\eta_{2})$, $\varsigma \in [0,1]$, $\forall \eta_{1},\eta_{2} \in \Theta_{ij}$;
		\item[\textit{(iii)}] $p_{ij}$ is Lipschitz continuous and differentiable w.r.t. $\theta$, implying $\exists K_{ij} \geq 0: ~|p_{ij}^{\prime}(\theta_{ij})| \leq K_{ij}$, $\forall (i,j) \in \Emc$.
	\end{enumerate} 
\end{asm}

Also, to provide analytical guarantees to multiple attacks striking the network, we will make use of the global quantity 
\begin{equation}\label{eq:K_Delta}
	K_{\Delta} = \max_{(u,v) \in \Emc_{\Delta}} \left\lbrace K_{uv}\right\rbrace. 
\end{equation}

\subsection{Guarantees for multiple attacks}\label{ssec:continuoustime_multi}
The guarantees in \cite[Theorem IV.1]{FabrisZelazo2022} can be extended as follows for a continuous-time multiple-attack scenarios. 

\begin{thm} \label{thm:maximum_perturb_codeword_multi}
	Assume that the characterization for objective decoding functions $p_i$ in Asm. \ref{asm:3properties} holds. 
	For a structured injection attack $\delta^{\theta} \in \BDelta^{\theta}$ affecting all edges in $\Emc_{\Delta}$ define $\Rmc_{\Emc_{\Delta}}(\Gmc)$ and $K_{\Delta}$ as in \eqref{eq:uncguar_multi} and \eqref{eq:K_Delta}, respectively. Then the perturbed consensus protocol \eqref{eq:perturbedSBDC} is stable and achieves agreement for all $\delta^\theta$ whenever
	\begin{equation}\label{eq:sgrepatt_multi}
		\left\|\delta^\theta\right\|_{\infty} < \rho^\theta_{\Delta} = (K_{\Delta} \Rmc_{\Emc_{\Delta}}(\Gmc) )^{-1},
	\end{equation}
	independently from the values taken by any codeword $\theta \in \Theta$. 
\end{thm}
\begin{proof}
	Similarly to the single edge attack case, Asm. \ref{asm:3properties} brings to each ordered logical step to conclude the thesis through Lem. \ref{thm:effective_resistance_m}. 
	Indeed, \textit{(i)-(iii)} lead to the determination of quantity
	$K_{ij}|\delta^{\theta}_{ij}|$, which can be seen as the maximum magnitude of an additive perturbation $\delta^{w}_{ij} := p_{ij}(\theta_{ij} + \delta^{\theta}_{ij})-p_{ij}(\theta_{ij})$ 
	affecting each $p_{ij}$, $\forall (i,j) \in \Emc$, independently from the transmitted codeword $\theta$. Consequently, the fact that $|\delta_{uv}^{w}| \leq K_{uv}|\delta_{uv}^{\theta}|$ holds for each edge $(u,v) \in \Emc_{\Delta}$ implies that the following chain of inequalities is verified:
	\begin{equation}\label{eq:double_conservatism}
		\left\|\delta^{w}\right\|_{\infty} \leq   \max_{(u,v) \in \Emc_{\Delta}} \left\lbrace K_{uv} |\delta^{\theta}_{uv}| \right\rbrace \leq K_{\Delta} \left\| \delta^{\theta} \right\|_{\infty}.
	\end{equation}
	Therefore, imposing inequality $K_{\Delta} \left\| \delta^{\theta} \right\|_{\infty} < \Rmc^{-1}_{\Emc_{\Delta}}(\Gmc) $, in accordance with \eqref{eq:uncguar_multi}, leads to the thesis.
\end{proof}

We notice that the leftmost inequality in \eqref{eq:double_conservatism} translates into an essential conservatism on the attack magnitude w.r.t. \eqref{eq:uncguar_multi}, similarly to the case $|\Emc_{\Delta}|=1$. This can be modulated by choosing constants $K_{ij}$ properly. 
Nonetheless, as soon as $|\Emc_{\Delta}|>1$ holds, another essential conservatism that hinges on the attack's scale $|\Emc_{\Delta}|$ may arise w.r.t. its one-dimensional version. In the sequel, this concept is referred to as the $\Emc_{\Delta}$-\textit{gap}, or \textit{resilience gap}, since it depends to the fact that, in general, one has $\Rmc_{\Emc_{\Delta}}(\Gmc) \geq \max_{(u,v)\in\Emc_{\Delta}} \left\lbrace  \Rmc_{\left\lbrace (u,v)\right\rbrace}(\Gmc) \right\rbrace$ (see discussion in Sec. \ref{conserv_analysis} for more details). 
	
Concerning, instead, the discrete-time guarantees for system \eqref{eq:perturbedSBDCdt} provided in \cite[Theorem VI.1, Corollary VI.1]{FabrisZelazo2022}, the following generalization can be made. 
\begin{cor} \label{thm:maximum_perturb_codeword_dt_multi}
	Assume that the characterization for objective decoding functions $p_i$ in Asm. \ref{asm:3properties} holds.
	Denote respectively with $\bar{w}_{i} = \sum_{j \in \Nmc_{i}} |w_{ij}|$ and $\Psi_{\Gmc} = \max_{i=1,\ldots,n} \left\lbrace \bar{w}_{i} \right\rbrace$ the weighted degree of the $i$-th node and the maximum weighted degree of the underlying graph $\Gmc$.
	Let a structured injection attack $\delta^{\theta} \in \BDelta^{\theta}$ affect all edges in $\Emc_{\Delta}$ and define
	\begin{align}
		\psi_{i}(\delta^\theta_{uv}) &= \bar{w}_{i}+K_{uv}|\delta^\theta_{uv}|, \quad \forall (u,v) \in \Emc_{\Delta}, ~i=u,v. \label{eq:psis} 
	\end{align}
	Then the perturbed consensus protocol \eqref{eq:perturbedSBDCdt} is stable for all $\delta^\theta$ such that both \eqref{eq:sgrepatt_multi} and
	\begin{equation}\label{eq:sgrepatt_dt_multi} 
		\phi_{\Gmc}(\delta^\theta) := \max \left\lbrace  \Psi_{\Gmc}, \max_{i\in\left\lbrace u,v \right\rbrace, (u,v)\in \Emc_{\Delta}} \psi_{i}(\delta^\theta_{uv}) \right\rbrace   <  \epsilon^{-1}
	\end{equation}
	hold for any fixed $\epsilon$, independently from the values taken by any codeword $\theta \in \Theta$. 
\end{cor}
\begin{proof}
	The proof is a generalization of that for Theorem VI.1 in \cite{FabrisZelazo2022}, with the only difference that $\phi_{\Gmc}(\delta^\theta)$ in \eqref{eq:sgrepatt_dt_multi} is harder to be computed w.r.t. its single-edge-attack version $\phi_{\Gmc}(\delta^\theta_{uv})$ because the combinatorial search space for the maximization of quantities $\psi_{i}(\delta^\theta_{uv})$ defined in \eqref{eq:psis} is, in general, larger in this kind of multiple-attack scenario.
\end{proof}

\begin{cor} \label{cor:maximum_perturb_codeword_dt_multi}
	Under all the assumptions adopted in Cor. \ref{thm:maximum_perturb_codeword_dt_multi}, defining $K_{\Delta}$ as in \eqref{eq:K_Delta}, $\Rmc_{\Emc_{\Delta}}(\Gmc)$ as in \eqref{eq:uncguar_multi} and setting $\epsilon < \Psi_{\Gmc}^{-1}$, the perturbed consensus protocol \eqref{eq:perturbedSBDCdt} is stable for all $\delta^\theta$ such that
	\begin{equation}\label{eq:sgrepatt_dt_simple_multi}
		\left\|\delta^{\theta}\right\|_{\infty} < \rho^{\theta}_{\Delta} =  K_{\Delta}^{-1} \min \{ \Rmc_{\Emc_{\Delta}}^{-1}(\Gmc) , (\epsilon^{-1}-\Psi_{\Gmc}) \}
	\end{equation}
	independently from the values taken by any codeword $\theta \in \Theta$. In particular, if $\epsilon$ is selected as follows
	\begin{equation}\label{eq:epsguar_rs_multi}
		\epsilon \leq \epsilon_{\Delta}^{\star} := (\Psi_{\Gmc} + \Rmc_{\Emc_{\Delta}}^{-1}(\Gmc))^{-1}
	\end{equation}
	then $ \rho^{\theta}_{\Delta} $ in \eqref{eq:sgrepatt_dt_simple_multi} is maximized as $\epsilon$ varies and condition \eqref{eq:sgrepatt_multi} needs to be fulfilled solely to guarantee robust stability.
\end{cor}

\begin{proof}
	Relation in \eqref{eq:sgrepatt_dt_simple_multi} is the combined result of guarantee in \eqref{eq:sgrepatt_multi} and that one obtainable imposing $\Psi_{\Gmc} + K_{\Delta}\left\|\delta^{\theta}\right\|_{\infty} < \epsilon^{-1}$ to satisfy \eqref{eq:sgrepatt_dt_multi}, since $\phi_{\Gmc}(\delta^\theta)$ can be upper bounded as $\phi_{\Gmc}(\delta^\theta) \leq \Psi_{\Gmc} + K_{\Delta}\left\|\delta^{\theta}\right\|_{\infty}$. On the other hand, relation \eqref{eq:epsguar_rs_multi} is derived enforcing $\Rmc_{\Emc_{\Delta}}^{-1}(\Gmc) \leq \epsilon^{-1}-\Psi_{\Gmc}$ with the purpose to maximize $\rho_{\Delta}^{\theta}$ as $\epsilon$ varies. 
\end{proof}

We conclude this subsection with the following remark.
\begin{rmk}
	Observe that no additional assumptions on the decoding functions $p_{i}$ w.r.t. \textit{(i)-(iii)} given in Subsec. \ref{ssec:modelchanneltampmulti} are required to solve Problem \ref{prob:resilenceSBDC}, i.e., to generalize the guarantees already yielded in \cite{FabrisZelazo2022} to a multiple-attack scenario.
\end{rmk}

\subsection{Analysis of the resilience gap}\label{conserv_analysis}

In light of the previous theoretical results, we discuss here how decoding functions can be seen as a useful tool to compensate against the resilience gap (the $\Emc_{\Delta}$-\textit{gap}).
This analysis starts by recalling the following preliminary proposition.

\begin{prop}[{\cite[Proposition V.3]{ZelazoBurger2017}}] \label{prop:V3}
	$~$\\ For any weighted undirected graph $\Gmc$ it holds that
	\begin{equation}\label{eq:V3_inequalities} 
	  	\Rmc_{\Emc_{\Delta}}^{\star}(\Gmc) \leq 
	  	\Rmc_{\Emc_{\Delta}}(\Gmc) \leq 
	  	\Rmc_{\Emc_{\Delta}}^{tot}(\Gmc),
	\end{equation}
	where $\Rmc_{\Emc_{\Delta}}(\Gmc)$ is defined as in \eqref{eq:uncguar_multi},
	\begin{equation}\label{eq:V3_r_one_edge}
		\Rmc_{\Emc_{\Delta}}^{\star}(\Gmc) = \underset{(u,v) \in \Emc_{\Delta}}{\max} \left\lbrace \Rmc_{\left\lbrace (u,v) \right\rbrace}(\Gmc) \right\rbrace
	\end{equation}
	and $\Rmc_{\Emc_{\Delta}}^{tot}(\Gmc) = \mathrm{tr} ( P^{\top} R_{(\Tmc,\Cmc)}^{\top} (R_{(\Tmc,\Cmc)} W R_{(\Tmc,\Cmc)}^{\top})^{-1} R_{(\Tmc,\Cmc)} P   )$.
\end{prop}

Prop. \ref{prop:V3} suggests us to define the $\Emc_{\Delta}$-\textit{gap} as
\begin{equation}\label{eq:Edeltagap}
	g(\Gmc,\Emc_{\Delta}) = 1-\Rmc_{\Emc_{\Delta}}^{\star}(\Gmc) / \Rmc_{\Emc_{\Delta}}(\Gmc) ,
\end{equation} 
with $\Rmc_{\Emc_{\Delta}}^{\star}(\Gmc)$ and $\Rmc_{\Emc_{\Delta}}(\Gmc)$ defined by \eqref{eq:V3_r_one_edge} and \eqref{eq:uncguar_multi}.
Indeed, the emerging conservatism related to the fact that multiple edges may be under attack ($|\Emc_{\Delta}|>1$) grows as the value of $g(\Gmc,\Emc_{\Delta}) \in  [0,1)$ increases.
Also, exploiting \eqref{eq:Edeltagap}, inequality \eqref{eq:sgrepatt_multi} can be rewritten as 
\begin{equation}\label{eq:sgrepatt_mono_multi}
	\left\| \delta^{\theta}\right\|_{\infty} < (1-g(\Gmc, \Emc_{\Delta})) /(K_{\Delta} \Rmc^{\star}_{\Emc_{\Delta}}(\Gmc) ),
\end{equation}
so that the quantity $\rho_{\Delta^{\star}}^{\theta} := (K_{\Delta} \Rmc^{\star}_{\Emc_{\Delta}}(\Gmc) )^{-1}$ can be seen as the maximum value taken by $\rho_{\Delta}^{\theta}$ in \eqref{eq:sgrepatt_multi} as $\Gmc$ and $\Emc_{\Delta}$ vary. 
In particular, if $g(\Gmc,\Emc_{\Delta})=0$ holds (indicating the absence of this kind of conservatism), there exists an edge $(u^{\star},v^{\star})\in \Emc_{\Delta}$ for which, by \eqref{eq:sgrepatt_mono_multi}, the network $\Gmc$ is robustly stable for all perturbations $\delta^{\theta}$ such that 
\begin{equation}\label{eq:sgrepatt_mono}
	\left\| \delta^{\theta}\right\|_{\infty}   < (K_{\Delta}  \Rmc_{\left\lbrace (u^{\star},v^{\star})\right\rbrace}(\Gmc))^{-1} = \rho_{\Delta^{\star}}^{\theta}.
\end{equation}
Remarkably, inequality \eqref{eq:sgrepatt_mono} expresses the same guarantee provided by \eqref{eq:sgrepatt_multi} as if $\Emc_{\Delta}=\left\lbrace (u^{\star},v^{\star}) \right\rbrace$ was assigned, similarly to the one-dimensional case $|\Emc_{\Delta}| = 1$ debated in \cite{FabrisZelazo2022}. Since the case where only one edge is attacked clearly represents a scenario in which the conservatism due to $|\Emc_{\Delta}| > 1$ is lacking, 
if the $\Emc_{\Delta}$-gap corresponding to a given setup $(\Gmc,\Emc_{\Delta})$ vanishes then the robustness of $\Gmc$ subject to perturbations $\delta^{\theta}$ striking the subset $\Emc_{\Delta}$ is maximized. 
The latter observation leads us to wonder which are all the possible circumstances where $g(\Gmc,\Emc_{\Delta}) = 0$ is satisfied? A partial yet extensive answer is given by the next proposition. 
\begin{prop} \label{prop:gap=0}
	Let $\Gmc=(\Vmc,\Emc,\Wmc)$, $\Emc_{\Delta} \subseteq \Emc$ the subset of perturbed edges in $\Emc$, and $g(\Gmc,\Emc_{\Delta})$ the $\Emc_{\Delta}$-gap  \eqref{eq:Edeltagap}. 
	Then  $g(\Gmc,\Emc_{\Delta}) = 0$ if either one of the following conditions hold:	i) $~|\Emc_{\Delta}| = 1$, or ii) $~2 \leq |\Emc_{\Delta}| \leq n-1 = |\Emc |$.
\end{prop}
\begin{proof}
	We prove each point of the statement separately.
	
	i) If $|\Emc_{\Delta}| = 1$ then there exists only one edge $(u^{\star},v^{\star})\in \Emc_{\Delta}$ for which the mid inequality in \eqref{eq:V3_inequalities} must hold tightly, since $\mathcal{R}^{\star}_{\Emc_\Delta}(\mathcal{G}) = \mathcal{R}_{\Emc_\Delta} (\mathcal{G})$. It follows that $g(\Gmc,\Emc_{\Delta}) = 0$.

	ii) In this case, since $\Gmc$ is connected by assumption then it has to be a spanning tree; consequently, it can be also denoted as $\Gmc=\Tmc(\Emc,\Vmc,\Wmc)$. Hence, one has $R_{(\Tmc,\Cmc)} = I_{n-1}$ and, therefore, formula \eqref{eq:uncguar_multi} yields
	\begin{equation}\label{eq:REDelta_acyclic}
		\Rmc_{\Emc_{\Delta}}(\Tmc) = \underset{k:~ e_{k} \in \Emc_{\Delta}}{\max} \left\lbrace | w_{k}^{-1} | \right\rbrace  , \quad \forall \Emc_{\Delta} \subseteq \Emc. 
	\end{equation}
	Observe that \eqref{eq:REDelta_acyclic} also holds for subsets $\Emc_{\Delta}$ such that $|\Emc_{\Delta}|=1$ and the $\max$ over $ \left\lbrace k:~ e_{k} \in \Emc_{\Delta} \right\rbrace$ in \eqref{eq:REDelta_acyclic} is taken exactly as in \eqref{eq:V3_r_one_edge}. We conclude that $g(\Gmc,\Emc_{\Delta}) = 0$ occurs even for $|\Emc_{\Delta}| \geq 2$ whenever $\Gmc=\Tmc(\Emc,\Vmc,\Wmc)$, as the mid inequality in \eqref{eq:V3_inequalities} is tightly satisfied under these conditions.
\end{proof}

To conclude, we note that the $\Emc_{\Delta}$-gap in \eqref{eq:sgrepatt_mono_multi} can be mitigated by an appropriate re-design procedure of   $K_{\Delta}$. 

\begin{prop}\label{prop:Kdeltaprime}
    Denote with $\rho_{\Delta^{\star}}^{\theta}(K_\Delta)$ the value of $\rho_{\Delta^{\star}}^{\theta}$ as $K_\Delta$ varies. Set $K_{\Delta}^{\prime} := (1-g(\Gmc,\Emc_\Delta)) K_{\Delta}$. Then 
    \begin{equation*}
        \rho_{\Delta^{\star}}^{\theta}(K_{\Delta}^{\prime}) > \rho_{\Delta^{\star}}^{\theta}(K_\Delta), \quad \forall K_{\Delta} > 0.
    \end{equation*}
\end{prop}
Prop. \ref{prop:Kdeltaprime} introduces an interesting tradeoff, where conservatism is reduced at the expense of shrinking the image of the decoding function (see also Fact IV.1 in \cite{FabrisZelazo2022}).  This leads to an increasing demand of the encryption capabilities.




\section{Numerical examples} \label{sec:simulations}
We here focus on a potential application of the proposed technique to semi-autonomous networks (SANs) consisting of leader-follower autonomous agents \cite{chapman2015semi}. Let us assume that subset $\Vmc_{l} \subseteq \Vmc$, $\Vmc_{l} \neq \emptyset$, collects all the leader agents of a SAN $\Smc = (\Gmc,\Umc)$, where $\Gmc = (\Vmc,\Emc,\Wmc)$ is an undirected and connected graph representing the network interactions of $\Smc$ and $\Umc = \left\lbrace \uu_{1}(t), \ldots,\uu_{|\Vmc_{l}|}(t) \right\rbrace$, $\uu_{\ell}(t) \in \mathbb{R}^{D}$, denotes the set of external inputs that directly influence each leader agents in $\Vmc_{l}$. Typically, setting $\mathbf{u}(t) = \vec{i \in \Vmc_l}{}{\mathbf{u}_{i}(t)}$, the continuous-time dynamics for SANs are yielded by
\begin{equation}\label{eq:SANct}
	\dot{\x} =  (L_{B}(\Gmc) \otimes I_{D})\x + (B \otimes I_{D}) \uu,
\end{equation}
where $L_{B}(\Gmc) = L(\Gmc) + \mathrm{diag}(B \ones_{|\Vmc_{l}|})$, with $B\in\mathbb{R}^{n \times |\Vmc_{l}|}$ such that $[B]_{i\ell} > 0$, if agent $i$ belongs to the leader set $\Vmc_{l}$; $[B]_{i\ell} = 0$, otherwise. Also, a discrete-time version of \eqref{eq:SANct} can be evidently provided by
\begin{equation}\label{eq:SANdt}
	\x(t+\rev{1}) =  (I_{N} - \epsilon(L_{B}(\Gmc) \otimes I_{D}))\x(t) + (B \otimes I_{D}) \uu(t),
\end{equation}
where the quantity $\epsilon$ preserves the same meaning of \eqref{eq:perturbedSBDCdt}.
The stability of SANs endowed with the discussed security mechanisms can be analyzed by observing that the positive semi-definiteness of $L(\Gmc)$ clearly implies the positive semi-definiteness of $L_{B}(\Gmc)$. Consequently, it is sufficient to ensure Thm. \ref{thm:maximum_perturb_codeword_multi} and Cor. \ref{thm:maximum_perturb_codeword_dt_multi} to hold in order to respectively guarantee the robust stability of protocols \eqref{eq:SANct} and \eqref{eq:SANdt}.

\begin{figure}[b!]
	\centering 
	\subfigure[]{\includegraphics[height=0.169\textwidth, trim={0cm -1cm 0cm 0cm},clip]{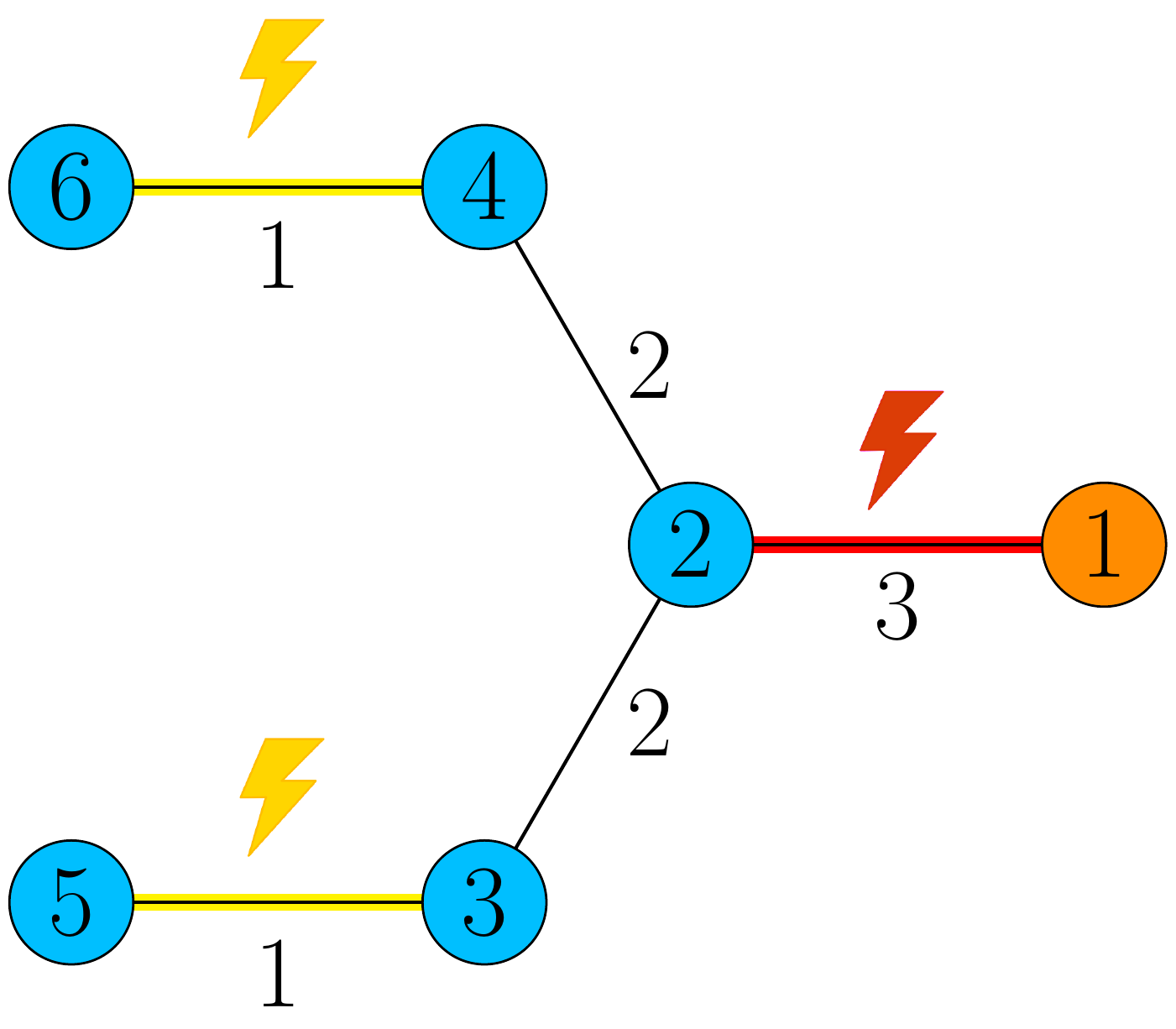}\label{fig:Graph_multi}}
	\hspace{0.0cm}
	\subfigure[]{\includegraphics[height=0.159\textwidth,trim={3cm 0.1cm 6cm 1cm},clip]{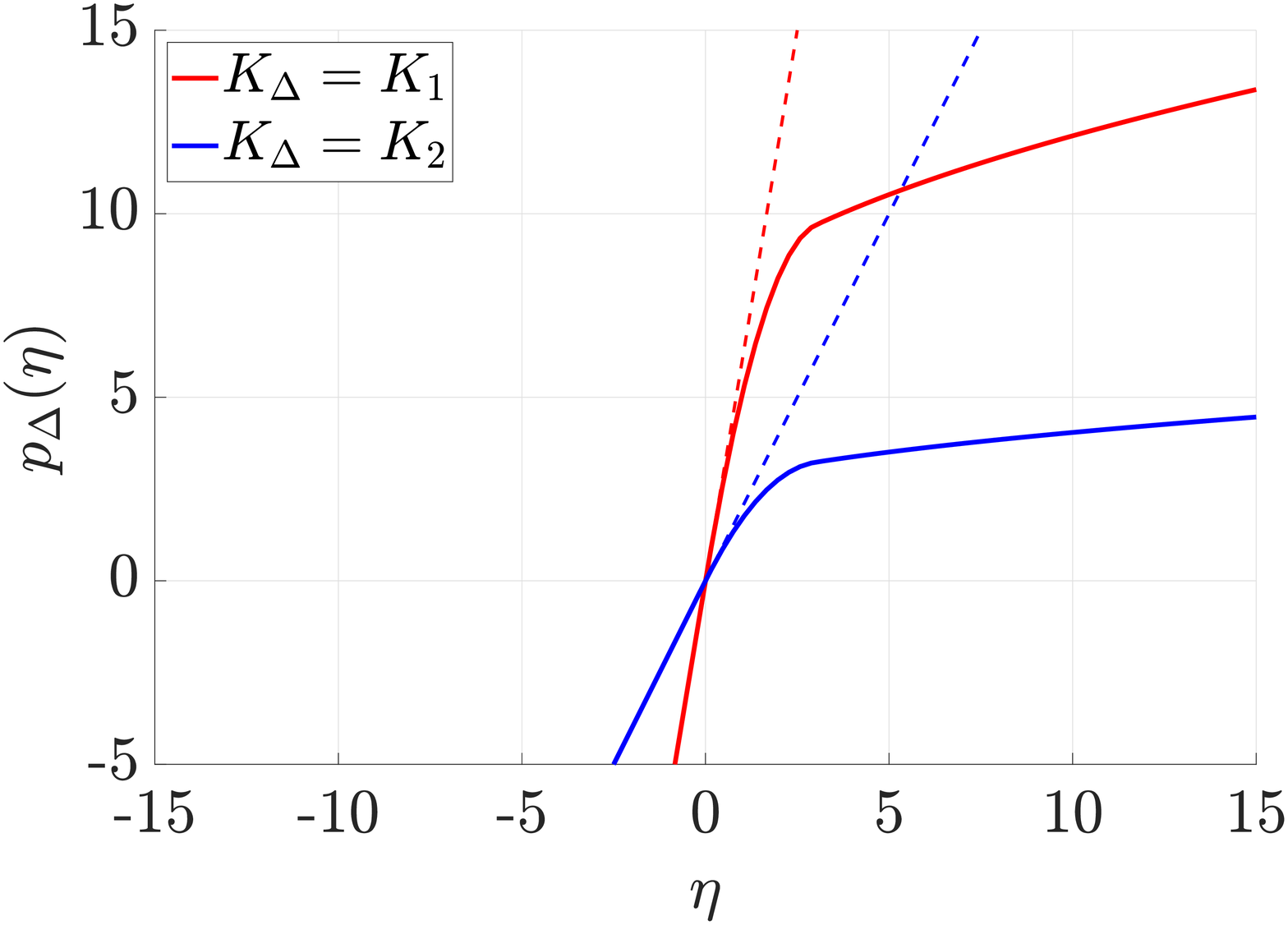}\label{fig:p}}
	\caption{(a): underlying topology of the given SAN (leaders in orange and followers in light blue); (b): Objective decoding functions considered for this case study.}
	\label{fig:sim_setup}
\end{figure}

In the following lines, we examine the SAN $\Smc = (\Gmc,\Umc)$ whose interconnections are modeled through the graph $\Gmc = (\Vmc,\Emc,\Wmc)$ illustrated in Fig. \ref{fig:Graph_multi}: the sole leader (node $1$) is marked in orange, the followers (remaining nodes) are marked in light blue, and the edge weights $\Wmc = \left\lbrace w_{k}\right\rbrace_{k=1}^{5}$ given by $w_{1}=w_{12}$, $w_{2}=w_{35}$, $w_{3}=w_{46}$, $w_{4}=w_{24}$, $w_{5}=w_{23}$ are chosen according to the same picture. Moreover, we set $D=1$, $[B]_{i\ell} =1$ for $\ell=1$ if $i=1$, and $\Umc = \left\lbrace u_{1} \right\rbrace$, with $u_{1} = -0.5$. We test the proposed security methods both in continuous and discrete time by uniformly adopting for all the edges the following objective decoding function
\begin{equation}\label{eq:dec_fun_gamma0}
	p_{\Delta}(\eta) = \begin{cases}
		K_{\Delta} \left(\frac{4}{13} \sqrt{\eta+1}+1\right), &\quad \text{if } \eta \geq 3; \\
		K_{\Delta} \left(-\frac{2}{13}\eta^{2}+\eta\right), &\quad \text{if } 0 \leq \eta < 3; \\
		K_{\Delta} \eta, & \quad  \text{if }\eta < 0;
	\end{cases}
\end{equation}
where $K_{\Delta}>0$ is a suitable Lipschitz constant for $p_{\Delta}(\eta)$, to be selected in order to ensure that $\left\|\delta^{w}\right\|_{\infty} \leq \rho_{\Delta}^{\theta} = 0.5$. More in detail, we consider two kinds of malicious attacks hitting part of $\Wmc$. The first external perturbation is delivered against edge $(1,2)$ (red strike in Fig. \ref{fig:Graph_multi}). Whereas, the second one involves not only edge $(1,2)$ but also edges $(3,5)$ and $(4,6)$ (yellow strikes in Fig. \ref{fig:Graph_multi}). Hence, under this specific adversarial setup, we denote the set of edges under attack as $\Emc_{1} = \left\lbrace (1,2) \right\rbrace$ and $\Emc_{2} = \left\lbrace (1,2), (3,5), (4,6) \right\rbrace$ for the first and second kind of perturbation, respectively.\\
Observing that $\Gmc$ in Fig. \ref{fig:Graph_multi} is an undirected tree (i.e. $\Gmc = \Tmc$) then $\Rmc_{\Emc_{\Delta}}(\Gmc)$ can be computed as in \eqref{eq:REDelta_acyclic}.
One can set $\epsilon = (\Rmc_{\Emc_{1}}^{-1}(\Gmc) + \Psi_{\Gmc})^{-1} = (2w_{1}+w_{4}+w_{5})^{-1} = 0.1$ in order to satisfy \eqref{eq:epsguar_rs_multi} for both $\Emc_{1}$, $\Emc_{2}$; thus, the Lipschitz constant $K_{\Delta}$ of $p_{\Delta}$ can be finally designed 
according to \eqref{eq:REDelta_acyclic} and \eqref{eq:sgrepatt_multi} so that it respectively takes values $K_{1} = 6$ and $K_{2} = 2$ for perturbations on edges $\Emc_{1}$ and $\Emc_{2}$. In Fig. \ref{fig:p}, the decoding function \eqref{eq:dec_fun_gamma0} is then depicted (solid curves) for such values of $K_{\Delta}$ along with the corresponding (dashed) lines having slope equal to its Lipschitz constant.

\vspace{-0.4mm}
We now discuss the robust stability of the SAN under investigation. 
In particular, we focus on the  trajectories of \eqref{eq:SANct} and \eqref{eq:SANdt} for $\Smc$ subject to the following perturbations:
$\delta_{1}^{\theta} = -0.5\rho_{\Delta}^{\theta} [1~\!\! 0~\!\! 0~\!\! 0~\!\! 0]^{\top}$, 
$\delta_{2}^{\theta} = -0.5\rho_{\Delta}^{\theta} [1~\!\! 1~\!\! 1~\!\! 0~\!\! 0]^{\top}$, 
where $\delta_{1}^{\theta}$ 
strikes the edge in $\Emc_{1}$ and $\delta_{2}^{\theta}$
strikes the three edges in $\Emc_{2}$. 
It is crucial to note that if $K_{\Delta} \!\!=\!\! K_{1}$ is selected as in the single-case attack scenario proposed in \cite{FabrisZelazo2022} adopting $p_{\Delta}$ to counter the perturbation $\delta_{2}^{\theta}$, then no stability guarantees can be given (see Figs. \ref{fig:ct_div}-\ref{fig:dt_div}). 
Indeed, despite this choice allows to mitigate the malicious effects of $\delta_{1}^{\theta}$, the additional perturbations on edges $(3,5)$ and $(4,6)$ are capable of disrupting the network state convergence towards $u_{1}$ as $t \!\rightarrow \!+\infty$ because $K_{\Delta} = K_{1}$ does not match condition \eqref{eq:sgrepatt_multi}. 
Instead, Figs. \ref{fig:ct_conv}-\ref{fig:dt_conv} show that such an issue is overtaken if $K_{\Delta} = K_{2}$ is set ensuring proper state convergence in accordance with \eqref{eq:sgrepatt_multi}.  


\begin{figure}[b!]\vspace{-.5cm}
	\vspace{-2.7mm}
	\centering
	\subfigure[Divergence (continuous time).
	]{\includegraphics[height=0.209\textwidth, trim={9cm 0.1cm 10cm 1cm},clip]{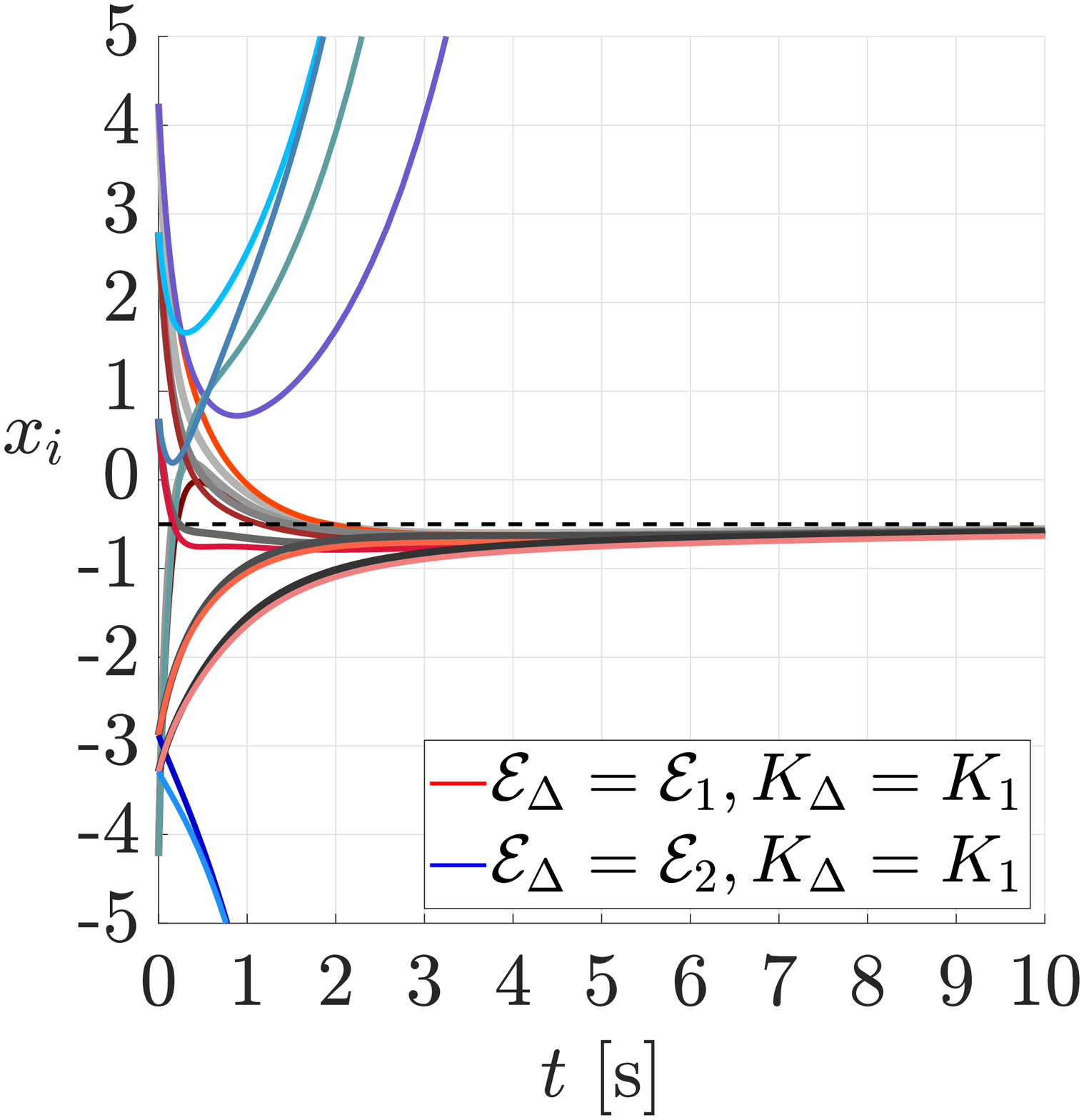}\label{fig:ct_div}}
	\hspace{0.0cm}
	\subfigure[Divergence (discrete time).
	]{\includegraphics[height=0.209\textwidth,trim={9cm 0.1cm 10cm 1cm},clip
		]{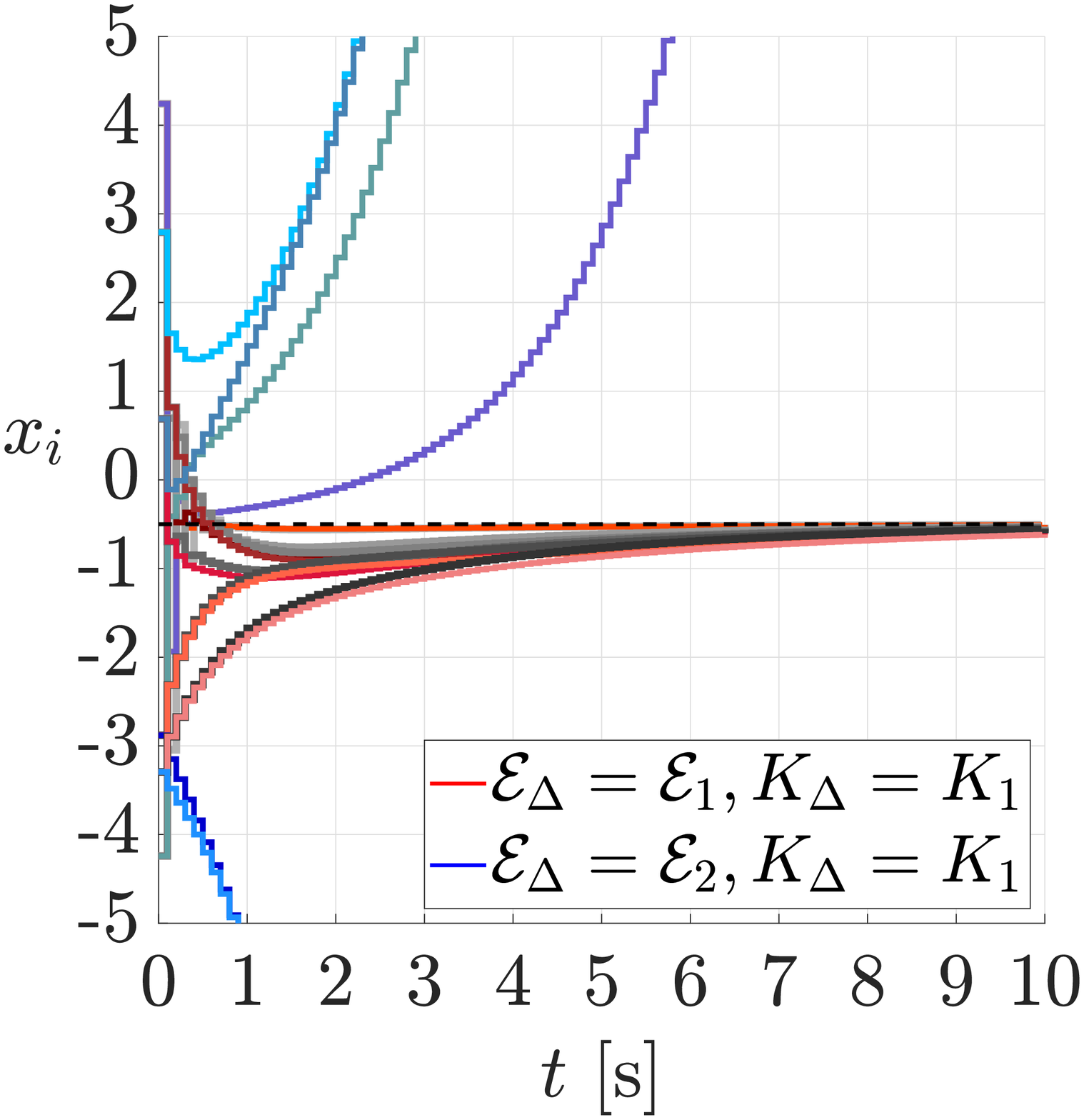}\label{fig:dt_div}}
	\\
	\vspace{-3.0mm}
	\subfigure[Convergence (continuous time).
	]{\includegraphics[height=0.219\textwidth, trim={9cm 0.1cm 10cm 1cm},clip]{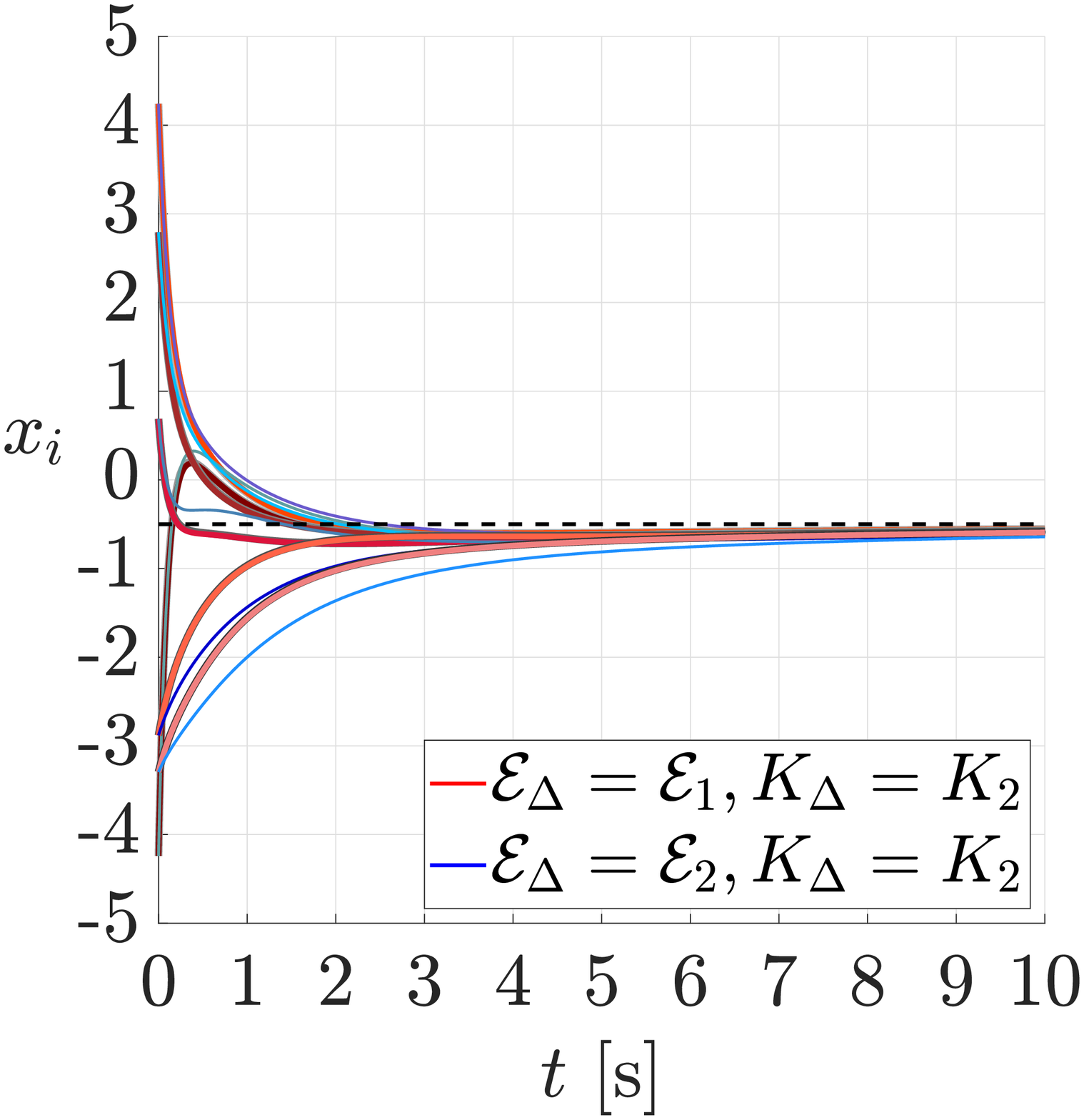}\label{fig:ct_conv}}
	\hspace{0.0cm}
	\subfigure[Convergence (discrete time).
	]{\includegraphics[height=0.219\textwidth, trim={9cm 0.1cm 10cm 1cm},clip]{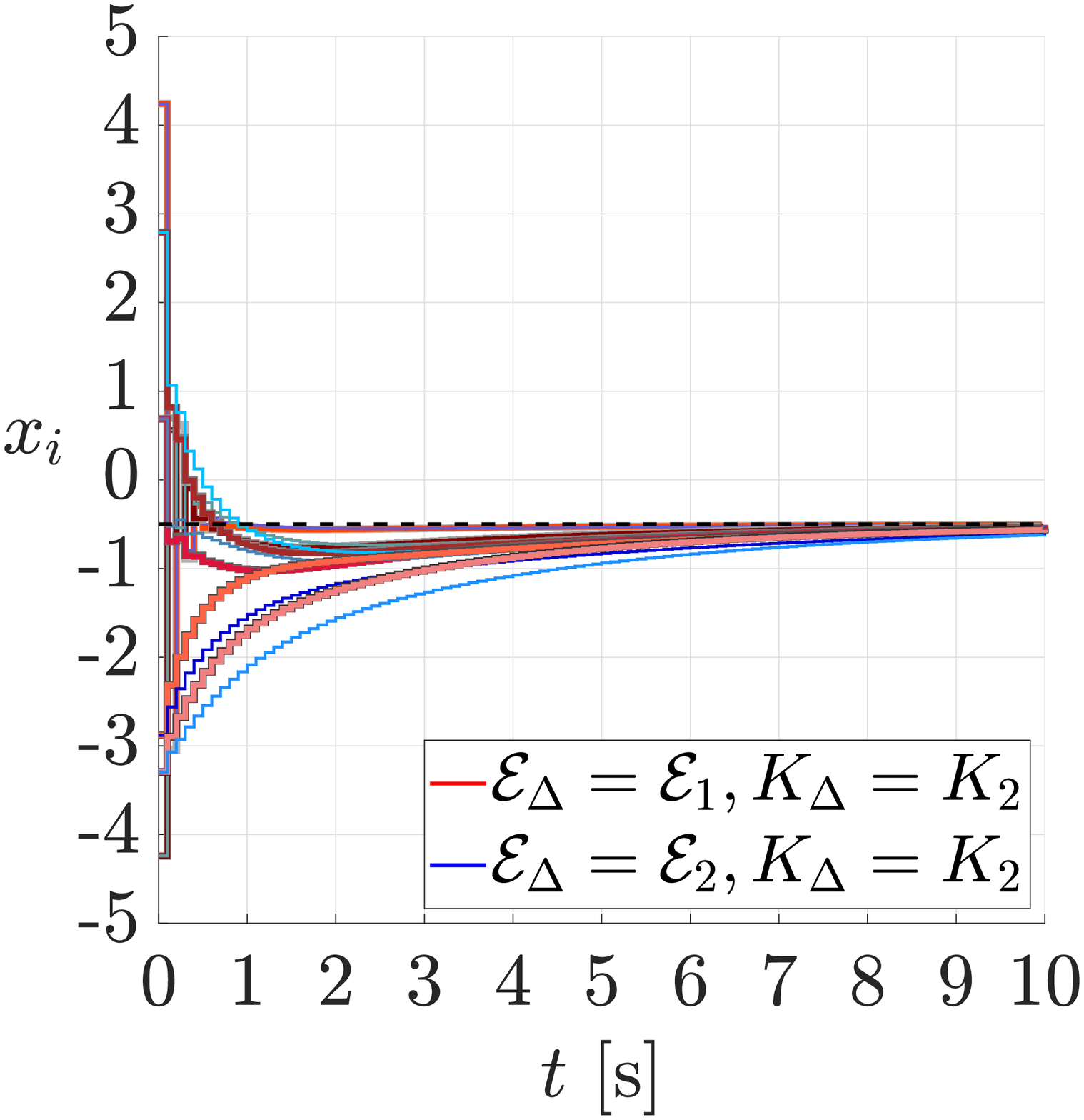}\label{fig:dt_conv}}
	\caption{
		Behaviors of the SAN trajectories as $\Emc_{\Delta}$ and $p_{\Delta}$ vary. According to $\epsilon$, each discrete time stamp is set to $0.1~s$. (a)-(b): stability guarantees for attack $\delta_{1}^{\theta}$ on $\Emc_{1}$ but not for $\delta_{2}^{\theta}$ on $\Emc_{2}$; (c)-(d): stability guarantees for the attacks $\delta_{i}^{\theta}$ on $\Emc_{i}$, $i=1,2$. 
	}
	\label{fig:matlab_sims}
\end{figure}

\vspace{-0.1mm}
In light of this, we can appreciate that the extension of the SBDC protocol towards multiple-attack scenarios requires, in general, more encryption capabilities and attention to the design phase in order to ensure robust stability compared to the single-attack case. Indeed, in this simulation, we have seen that only by considering $\Rmc_{\Emc_{\Delta}}(\Gmc)$ computed as in \eqref{eq:uncguar_multi} and lowering the value for $K_{\Delta}$ the agreement can be reached.

\section{Final remarks and future directions}\label{sec:conclusions}
The paper broadens the approach devised in \cite{FabrisZelazo2022} for secure consensus networks to a multiple attack scenario, both in the continuous and discrete time domains. Even in this framework, small-gain-theorem-based stability guarantees are yielded to this aim and, remarkably, no additional assumption is needed to provide such a generalization. In addition, the conservatism arising from a multiplicity of threats has been addressed and analyzed.
Future works will involve new developments towards other multi-agent protocols, such as nonlinear consensus and distance-based formation control.


\bibliographystyle{IEEEtran}
\bibliography{biblio}

\end{document}